\documentclass[10pt]{article}
\usepackage{amsfonts}
\usepackage{amssymb,amsmath,amsthm,latexsym}
\usepackage{amsmath,amsthm,amssymb}
\newtheorem{Theorem}{Theorem}[section]
\newtheorem{lem}[Theorem]{Lemma}

\newtheorem{Definition}[Theorem]{Definition}
\newtheorem{Corollary}[Theorem]{Corollary}

\numberwithin{equation}{section}

\begin{document}
\title{Repeated-root constacyclic codes of length $2\ell^mp^n$\footnote{
 E-Mail addresses: bocong\_chen@yahoo.com (B. Chen),
 hdinh@kent.edu (H. Q. Dinh),
  hwliu@mail.ccnu.edu.cn (H. Liu). }}

\author{Bocong Chen$^{1,3}$,~Hai Q. Dinh$^2$,~ Hongwei Liu$^1$}

\date{\small
${}^1$School of Mathematics and Statistics,
Central China Normal University,
Wuhan, Hubei, 430079, China\\
${}^2$Department of Mathematical Sciences,
Kent State University,
4314 Mahoning Avenue, Warren, OH 44483, USA\\
${}^3$School of Physical \& Mathematical Sciences,
         Nanyang Technological University, Singapore 637616, Singapore\\}

\maketitle

\begin{abstract}
For any different odd primes $\ell$ and $p$, structure of constacyclic codes of length $2\ell^mp^n$ over a finite field $\mathbb F_q$ of characteritic $p$ and their duals is established in term of their generator polynomials. Among other results, all linear complimentary dual and self-dual constacyclic codes of length $2\ell^mp^n$ over $\mathbb F_q$ are obtained.


\end{abstract}

\section{Introduction}
Constacyclic codes over finite fields form a remarkable class of linear codes,
as they  include the important family of cyclic codes. In fact, the class of cyclic codes is
one of the most significant and well studied of  all codes. Many well known codes, such as BCH, Kerdock, Golay, Reed-Muller, Preparata, Justesen and binary Hamming codes, are either cyclic codes or can be constructed from cyclic codes.
Constacyclic codes also have practical applications as
they can be efficiently encoded using simple shift registers. They have rich algebraic structures for efficient error detection and correction, which explains their preferred role in engineering.

Let $\mathbb F_q$ be the finite field of order $q$, where $q$ is a power of a prime $p$.
Given a nonzero element
$\lambda$ of $\mathbb F_q$, $\lambda$-constacyclic codes of
length $n$ over $\mathbb F_q$ are classified as the ideals $\langle g(X) \rangle$ of
the quotient ring $\mathbb F_q[X]/\langle X^n-\lambda\rangle$, where the generator polynomial $g(X)$ is the unique monic polynimial of minimum degree in the code, which is a divisor of $X^n-\lambda$.

Obviously, there are $q-1$ classes constacyclic codes of length $n$ over $\mathbb{F}_q$. However, it turned out that many of them are equivalent in the sense that they have same structures. Thus, the natural question, that under what conditions on $\lambda$ and $\mu$ such that $\lambda$-constacyclic
codes of length $n$ and $\mu$-constacyclic codes of length $n$ have the same algebraic structures, has been studied by many authors.
Particular cases of this question  have been considered since the late 1990s, even for the more general alphabets of finite rings.
 Wolfmann \cite{Wolfmann} showed that cyclic and negacyclic codes over $\mathbb Z_4$, the ring of integers modulo $4$, have the same structure for odd code lengths. Dinh and L\'opez-Permouth \cite{Dinh2004} generalized that to obtain that this fact holds  true for cyclic and negacyclic codes of odd lengths over any finite chain ring. When the lengths are a prime power, say $p^s$, Dinh \cite{Dinh2010} showed that all constacyclic codes over the finite field $\mathbb F_{q}$  have the same structure; and over the chain ring $\mathbb F_{q}+u\mathbb F_{q}$, the author gave the classification that all $(\alpha+u\beta)$-constacyclic codes have the same structures, and all $\gamma$-constacyclic codes are equivalent, for arbitrary nonzero elements $\alpha, \beta, \gamma$ of the field $\mathbb F_{q}$.

Recently, we introduced an equivalence relation $``\sim_n"$ called $n$-equivalence for the nonzero
elements of $\mathbb F_q$ to classify  constacyclic codes  of length $n$ over $\mathbb F_q$ such that the constacyclic codes belonging to
the same equivalence class have the same distance structures and the same algebraic structures \cite{CDL14}.

\begin{Definition}\label{n-equivalence}
Let $n$ be a positive integer. For any  elements $\lambda$, $\mu$
of $\mathbb F_q^*$,  we say that $\lambda$ and $\mu$ are {\em $n$-equivalent} in $\mathbb F_q^*$
and denote by  $\lambda \sim_n\mu$ if the
polynomial $\lambda X^n-\mu$ has a root in $\mathbb F_q$.
\end{Definition}

We obtained that $\lambda$ and $\mu$ are $n$-equivalent if and only if they belong  to the same coset of $\langle\xi^n\rangle$ in $\langle\xi\rangle$.
That means the distinct  cosets of  $\langle\xi^n\rangle$ in $\langle\xi\rangle$
give all the $n$-equivalence classes, thus
each $n$-equivalence class contains the same number of elements. Moreover, we showed that, for any $\lambda,\mu\in \mathbb F_q^*$, $\lambda\sim_n\mu$
if and only if there exists an $a\in \mathbb F_q^*$ such that
$$\varphi:~\mathbb F_q[X]/\langle X^n-\mu\rangle~\rightarrow~\mathbb F_q[X]/\langle X^n-\lambda\rangle$$
$$f(X)\mapsto f(aX),$$
is a ring isomorphism, and hence, the generator polynomial
of the $\mu$-constacyclic code $C$ and the generator polynomial of
the $\lambda$-constacyclic code
$\varphi(C)$
are linked in a very simple way.

For any nonzero $\lambda \in \mathbb F_q$, $\mathbb F_q[X]/\langle X^n-\lambda\rangle$ is a principal ideal ring, i.e., every ideal of  $\mathbb F_q[X]/\langle X^n-\lambda\rangle$ can be generated by a monic divisor of $X^n-\lambda$. It follows that
the irreducible factorization of $X^n-\lambda$
in $\mathbb F_q[X]$ determines all $\lambda$-constacyclic codes of length $n$ over $\mathbb F_q$.
Most of the authors assume from the outset that the code length $n$ is coprime to  $q$. This condition implies that every root of $X^n-\lambda$  is a simple root in an extension field of $\mathbb F_q$, which provides a description of all such roots, and hence, $\lambda$-constacyclic codes, by cyclotomic cosets modulo $n$.
In contrast  to  simple-root codes, constacyclic codes with $p$ dividing $n$ are called repeated-root constacyclic codes,
which were first studied in 1967 by Berman \cite{Berman}, and then
by several authors such as Massey {\it et al.} \cite{Massey}, Falkner {\it et al.} \cite{FKHZ79}, Roth, Seroussi \cite{RS86} and Salagean \cite{Salagean}.
Repeated-root codes were first investigated in the most generality in the 1990s by Castagnoli {\it et al.} \cite{Castagnoli}, and van Lint \cite{van}, where they showed that repeated-root cyclic codes have a concatenated construction, and are not asymptotically good.
However, it turns out that optimal repeated-root constacyclic  codes still exist \cite{Dinh2008,Dinh2010,Kai10}.
In particular,
it has been proved that self-dual cyclic codes over a finite field exist
precisely  when the code length is even and the characteristic of the underlying field is two \cite{Kai,Jia}.
These motivate researchers to further study this class of codes.

Recently, Dinh, in a series of papers \cite{Dinh2012, Dinh2013, Dinh2014}, determined the generator polynomials of all  constacyclic
codes over $\mathbb F_{q}$, of lengths $2p^s$, $3p^s$ and $6p^s$.
Dual constacyclic codes of these lengths were also  discussed. These results have been extended to more general code lengths.
The generator polynomials of all
constacyclic codes of length $2^tp^s$ over $\mathbb F_{q}$ were given in \cite{Bakshi1},
where $q$ is a power of an odd prime $p$.
The generator polynomials
of all constacyclic codes of length $\ell p^s$ over $\mathbb F_{q}$ were characterized in \cite{Chen} and \cite{CDL14}, where $\ell$ is a prime different
from $p$. Moreover, \cite{CDL14} identified the duals of all such conctacyclic codes, and provided all self-dual and all linear complimentary dual constacylic codes.

In this paper, we continue to extend the main results of \cite{Dinh2012, Dinh2013, Dinh2014} to a more general
code length of $2\ell^m p^n$, for any different odd primes $\ell$ and $p$. According to the equivalence classes induced by $``\sim_{2\ell^m p^n}"$,
all  constacyclic codes   of length
$2\ell^m p^n$ over $\mathbb F_{q}$ and their duals are characterized in Sections 3 and 4, respectively.
As an application, all linear complimentary dual constacylic codes of length $2\ell^m p^n$ are obtained. Since it is known that self-dual constacyclic codes can only occur among the classes of cyclic and negacyclic codes, and self-dual cyclic codes  over $\mathbb F_q$ does not exist because $p$ is odd, it follows that
all self-dual constacyclic codes of length $2\ell^m p^n$ over $\mathbb F_q$ can only occur among the class of negacyclic codes. We provide all such self-dual negacyclic codes.

\section{Preliminaries}
Starting from this section till the end of this paper, $\mathbb F_q$ denotes the finite field of order   $q$, where $q$ is a power of an odd prime $p$.
Let  $\mathbb F_q^*=\mathbb{F}_q\setminus\{0\}$.
For $\beta\in \mathbb F_q^*$,  we denote by $\rm{ord}(\beta)$
the order of $\beta$ in the group $\mathbb F_q^*$;
then $\rm{ord}(\beta)$ is a divisor of $q-1$, and $\beta$
is called a {\em primitive $\rm{ord}(\beta)$th root of unity}.
It is  known that $\mathbb F_q^*$ is  generated by a primitive $(q-1)$th root
$\xi$ of unity, i.e.,  $\mathbb F_q^*=\langle\xi\rangle$.
As usual, for integers $a,b$ and a prime $l$,
$a\mid b$ means that  $a$ divides $b$,
$l^a\Vert b$ means that $l^a\mid b$ but $l^{a+1}\nmid b$.

Let $m$ be a  positive integer and $\ell$  an odd prime different from $p$.
Let $\mathbb{Z}_{\ell^m}=\{[b]_{\ell^m}\,|\,\hbox{$b$ is an integer}\}$
be  the ring consisting of all residue classes modulo  $\ell^m$ and
$\mathbb{Z}_{\ell^m}^*$ be the
unit group of the ring.
It is well known that $\mathbb{Z}_{\ell^m}^*$ is a cyclic group.  We denote  by $\langle q\rangle$,
the cyclic subgroup of $\mathbb{Z}_{\ell^m}^*$ generated by $[q]_{\ell^m}$.
Let  $\langle q\rangle$ act on $\mathbb{Z}_{\ell^m}$
by the following rule:
$$
q^i\cdot[b]_{\ell^m}=[bq^i]_{\ell^m}, ~~~~\hbox{for any  integer $i$ and $[b]_{\ell^m}\in \mathbb{Z}_{\ell^m}$}.
$$
For any integer $t$, the orbit of $[t]_{\ell^m}$,
$$
C_t=\Big\{t,tq,tq^2,\cdots\, tq^{m_t-1}\Big\}
$$
is called the {\em $q$-cyclotomic coset of $t$ modulo $\ell^m$}, where the
elements in the brace are calculated modulo $\ell^m$ and  $m_t$ is the cardinality of the orbit of
$[t]_{\ell^m}$. It is readily seen that $m_t$ is equal to the multiplicative order of $q$ modulo $\frac{\ell^m}{\gcd(\ell^m,t)}$.

We denote by ${\rm ord}_{\ell}(q)=f$, the multiplicative order of $q$ in $\mathbb{Z}_{\ell}^*$. Write
$$
q^f=1+\ell^st,~~ \ell\nmid t,~ s\geq1.
$$
For any integer $r$, $1\leq r\leq m$,
let
\begin{equation}\label{order}
\lambda(r):=f\ell^{\max(r-s,0)}.
\end{equation}
One knows that ${\rm ord}_{\ell^{r}}(q)=\lambda(r)$ (see \cite{Bakshi2} or \cite{Sharma}).
Let $\delta(r)=\frac{\phi(\ell^{r})}{\lambda(r)}$, where $\phi$ denotes
Euler's phi-function.
Let $g$ be a fixed generator of the cyclic group $\mathbb{Z}_{\ell^m}^*$. By \cite[Theorem 1]{Sharma},
$C_0=\{0\},$ and
$$
C_{\ell^{m-r}g^k}=\Big\{\ell^{m-r}g^k, \ell^{m-r}g^kq,\cdots,\ell^{m-r}g^kq^{\lambda(r)-1}\Big\},
~0\leq k\leq\delta(r)-1, ~1\leq r\leq m,
$$
consist all the distinct $q$-cyclotomic cosets modulo $\ell^m$.
For simplify, we write
$C_{\rho_0}=\{0\}$ and $C_{\rho_k}$, $1\leq k\leq e$ to denote
all the distinct $q$-cyclotomic cosets modulo $\ell^m$.
We then see that $e=\sum_{r=1}^{m}\delta(r)$.

Take $\eta$ to be a primitive $\ell^m$th root of unity
(maybe in an extension field of $\mathbb F_q$), and denote by $M_{\rho_k}(X)$,
the minimal polynomial of $\eta^{\rho_k}$ over $\mathbb F_q$.
It is well known that (e.g. see \cite[Theorem 4.1.1]{Huffman}):
\begin{equation}\label{simple-irr-decomposition}
X^{\ell^m}-1=M_{\rho_0}(X)M_{\rho_1}(X)M_{\rho_2}(X)
 \cdots M_{\rho_e}(X),
\end{equation}
with
$$
M_{\rho_0}(X)=X-1,~~~~
M_{\rho_k}(X)=\prod\limits_{s\in C_{\rho_k}}(X-\eta^s),~1\leq k\leq e,
$$
all being monic irreducible in $\mathbb F_q[X]$.

We need to  determine
the distinct $q^2$-cyclotomic cosets modulo $\ell^m$.
It  requires to consider two subcases.
If $f={\rm ord}_{\ell}(q)$ is odd,  namely $\lambda(r)={\rm ord}_{\ell^{r}}(q)$ is odd for each $1\leq r\leq m$,
then ${\rm ord}_{\ell^{r}}(q)={\rm ord}_{\ell^{r}}(q^2)$,
which means that the cyclic subgroup generated by $[q]_{\ell^{r}}$ in $\mathbb{Z}_{\ell^{r}}^*$ is equal to
the cyclic subgroup generated by $[q^2]_{\ell^{r}}$, i.e. $\langle q\rangle=\langle q^2\rangle$
in $\mathbb{Z}_{\ell^{r}}^*$;
in particular, $\langle q\rangle=\langle q^2\rangle$ in $\mathbb{Z}_{\ell^m}^*$.
By the definition of $q^2$-cyclotomic cosets, $C_{\rho_0}=\{0\}$ and $C_{\rho_k}$, $1\leq k\leq e$,
also consist all the distinct
$q^2$-cyclotomic cosets modulo $\ell^m$. It follows that Formula~(\ref{simple-irr-decomposition})
gives the irreducible factorization of $X^{\ell^m}-1$ in $\mathbb F_{q^2}[X]$.
If $f$ is even, we deduce that ${\rm ord}_{\ell^{r}}(q^2)=\frac{\lambda(r)}{2}$ for any $1\leq r\leq m$.
It is straightforward to verify that $D_{0}=\{0\}$,
$$
D_{\ell^{m-r}g^j}=\Big\{\ell^{m-r}g^j, \ell^{m-r}g^j\cdot q^2,\cdots, \ell^{m-r}g^j\cdot q^{2(\frac{\lambda(r)}{2}-1)}\Big\},
$$
and
$$
D_{\ell^{m-r}g^jq}=\Big\{\ell^{m-r}g^jq, \ell^{m-r}g^jq\cdot q^2,\cdots, \ell^{m-r}g^jq\cdot q^{2(\frac{\lambda(r)}{2}-1)}\Big\}
$$
consist all the distinct $q^2$-cyclotomic cosets modulo $\ell^m$, where $0\leq j\leq\delta(r)-1$ and $1\leq r\leq m$.
Observe that
$$
C_{\ell^{m-r}g^j}=D_{\ell^{m-r}g^j}\bigcup D_{\ell^{m-r}g^{j}q},
\hbox{~ for each $0\leq j\leq\delta(r)-1$~ and  ~$1\leq r\leq m$}.
$$
For simplify, we write
$D_{\rho_0}=\{0\}$,  $D_{\rho_k}$ and $D_{\rho_kq}$, $1\leq k\leq e$ such that $C_{\rho_k}=D_{\rho_k}\bigcup D_{\rho_kq}$, to denote
all the distinct $q^2$-cyclotomic cosets modulo $\ell^m$.
By~\cite[Theorem 4.1.1]{Huffman} again, we have
\begin{equation}\label{simple-irr-decomposition2}
X^{\ell^m}-1=(X-1) N_{\rho_1}(X)N_{\rho_1q}(X)N_{\rho_2}(X)N_{\rho_2q}(X)
 \cdots N_{\rho_e}(X)N_{\rho_eq}(X),
\end{equation}
with
$$
N_{\rho_k}(X)=\prod\limits_{s\in D_{\rho_k}}(X-\eta^s),~~N_{\rho_kq}(X)=\prod\limits_{t\in D_{\rho_kq}}(X-\eta^t), ~~1\leq k\leq e,
$$
all being monic irreducible in $\mathbb F_{q^2}[X]$.

In the rest of this section, we recall some basic concepts and results about constacyclic  codes over $\mathbb F_q$.
Let $\mathbb F_q^n$  be the $\mathbb F_q$-vector space of $n$-tuples.
A {\em linear code} $C$ of length $n$ over $\mathbb F_q$ is an $\mathbb F_q$-subspace of $\mathbb F_q^n$.
If $\lambda$ is a nonzero element of $\mathbb F_q$, a linear code $C$ of length $n$ over $\mathbb F_q$ is called
{\em $\lambda$-constacyclic} if $(\lambda c_{n-1}, c_0,\cdots,c_{n-2})\in C$
for every $(c_{0}, c_1,\cdots,c_{n-1})\in C$.  When $\lambda$ =$1$,
$\lambda$-constacyclic codes are just cyclic codes and when $\lambda=-1$,
$\lambda$-constacyclic codes are known as negacyclic codes.

For any $\lambda$-constacyclic  code $C$ of length $n$ over $\mathbb F_q$, the {\it dual code of $C$}
is defined as $C^\perp=\{u\in \mathbb F_q^n\,|\,u\cdot v=0, \mbox{for any $v\in C$}\}$,
where $u\cdot v$ denotes the standard Euclidean inner product of $u$ and $v$ in $\mathbb F_q^n$.
The code $C$ is said to be {\em self-orthogonal}
if $C\subseteq C^\perp$ and {\it self-dual} if $C=C^\perp$.
It turns out that   the dual of a $\lambda$-constacyclic code is a $\lambda^{-1}$-constacyclic code;
specifically,  the dual of a cyclic code is a cyclic code
and the dual of a negacyclic code is a negacyclic code (e.g. see \cite[Proposition 2.2.]{Dinh2012}).

We  know that  any $\lambda$-constacyclic  code $C$
of length $n$ over $\mathbb F_q$ is identified with exactly one ideal
of the quotient algebra $\mathbb F_q[X]/\langle X^n-\lambda\rangle$,
which is generated uniquely by a monic divisor $g(X)$ of $X^n-\lambda.$
In this case, $g(X)$ is called the {\em generator polynomial}
of $C$ and denote it by $C=\langle g(X)\rangle$.
Assume that $C=\langle g(X)\rangle$ is a $\lambda$-constacyclic  code of length $n$ over $\mathbb F_q$,
where $g(X)$ is the generator polynomial of $C$. Let $h(X)=\frac{X^n-\lambda}{g(X)}$.
It is known that its dual code
$C^\perp$ has generator polynomial $h^*(X)$, where $h^*(X)=h(0)^{-1}X^{\deg h}h(\frac{1}{X})$
is called  the {\it reciprocal polynomial} of $h(X)$.
Note that $h^*(X)$ is  a monic divisor of $X^n-\lambda^{-1}$.
If a polynomial is equal to its reciprocal polynomial, then it is called {\it self-reciprocal}.
Suppose that $f(X)\in \mathbb{F}_q[X]$ is a polynomial with leading coefficient
$a_n\neq0.$ We denote by $\hat f(X)$,  the monic  polynomial such that
$\hat f(X)=a_n^{-1}f(X)$.

\section{Constacyclic codes of length $2\ell^mp^n$ over $\mathbb F_q$}
Let $\ell$ be an odd prime different from $p$ as before.
Recall from (\ref{order}) that ${\rm ord}_{\ell^{r}}(q)=\lambda(r)$ and $\delta(r)=\frac{\phi(\ell^{r})}{\lambda(r)}$, $1\leq r\leq m$.
We take a primitive $\ell^m$th root $\eta$ of unity in the finite field $\mathbb F_{q^{\lambda(m)}}$; by (\ref{simple-irr-decomposition}),
we have the factorization of $X^{2\ell^m p^n}-1$ into irreducible factors over $\mathbb F_q$ as follows:
\begin{equation}\label{cyclic-decomposition}
X^{2\ell^m p^n}-1=(X^{2\ell^m}-1)^{p^n}=(X^{\ell^m}-1)^{p^n}(X^{\ell^m}+1)^{p^n}=
\prod\limits_{i=0}^{e}M_{\rho_i}(X)^{p^n}\hat M_{\rho_i}(-X)^{p^n}.
\end{equation}

The following lemma (proven in \cite{CDL14}) shows that in order to obtain all constacyclic codes of length $n$ over $\mathbb F_q$,
we only need to consider  $\lambda$-constacyclic codes, where $\lambda$ runs over
any fixed transversal of $\langle \xi^n\rangle$ in $\langle\xi\rangle$.
\begin{lem}\label{n-equivalence} For any $\lambda,\mu\in \mathbb F_q^*$,
the following four statements are equivalent:

{\bf(i)}\quad $\lambda^{-1}\mu\in\langle\xi^n\rangle$.

{\bf(ii)}\quad
$(\lambda^{-1}\mu)^d=1,$ where $d=\frac{q-1}{\gcd(n,q-1)}$.

{\bf(iii)}\quad $\lambda$ and $\mu$ are $n$-equivalent in $\mathbb F_q^*$, namely there exists an element
$a\in \mathbb F_{q}^*$ such that $a^n\lambda=\mu$.

{\bf(iv)}\quad There exists an $a\in \mathbb F_q^*$ such that
$$\varphi_a:~\mathbb F_q[X]/\langle X^n-\mu\rangle~\rightarrow~\mathbb F_q[X]/\langle X^n-\lambda\rangle$$
$$f(X)\mapsto f(aX),$$
is an $\mathbb F_q$-algebra isomorphism.

\noindent
In particular, the number of the $n$-equivalence classes
in $\mathbb F_q^*$ is equal to $\gcd(n, q-1).$
\end{lem}

If $\lambda$ and $\mu$ are $n$-equivalent,
we say that,  the $\lambda$-constacyclic codes of length $n$ are $n$-equivalent
to the $\mu$-constacyclic codes of length $n$.
That is,  it is enough to study the $n$-equivalence classes of constacyclic codes.

By Lemma~\ref{n-equivalence}, the number of  $2\ell^m p^n$-equivalence classes
in $\mathbb F_q^*$ is equal to $\gcd(2\ell^m p^n, q-1)=\gcd(2\ell^m, q-1)$.
Clearly, the cases $\gcd(\ell, q-1)=1$ and $\gcd(\ell, q-1)=\ell$ are distinguishable.

We first consider the case $\gcd(\ell, q-1)=1$.
In this situation, we have
$$
\mathbb F_q^*=\big\langle\xi\big\rangle=
\big\langle\xi^{2\ell^mp^n}\big\rangle\cup\xi^{p^n}\big\langle\xi^{2\ell^mp^n}\big\rangle,
$$
which means the $\lambda$-constacyclic codes are $2\ell^mp^n$-equivalent to
the cyclic codes or  $\xi^{p^n}$-constacyclic codes by Lemma~\ref{n-equivalence}.
Now we  take an element $\alpha_1$  in $\mathbb F_{q^2}$ satisfying $\alpha_1^2=\xi$.
We see that $\alpha_1\in \mathbb{F}_{q^2}$ and $\alpha_1\not\in \mathbb{F}_q$,
because $X^2-\xi\in \mathbb{F}_q[X]$ is irreducible.
It is readily seen that $\alpha_1\in \mathbb{F}_{q^2}$ is a primitive $2(q-1)$th root of unity.
Let
$S=\{\alpha\in\mathbb{F}_{q^2}^*\,|\, \alpha ~\hbox{is a primitive $2(q-1)$th root of unity}\}$.
It follows from $\gcd(\ell^m,2(q-1))=1$ that there is a  bijection
$\theta:~S\rightarrow S$  such that $\theta(\alpha)=\alpha^{\ell^m}$ for any $\alpha\in S$.
Thus, a unique element of $S$,  say $\beta_1$,  can be found such that $\alpha_1^{-1}=\theta(\beta_1)=\beta_1^{\ell^m}$, i.e.,
$\beta_1^{\ell^m}\alpha_1=1$.
Obviously, $\beta_1\not\in \mathbb{F}_q$.
We claim that $\beta_1^q=-\beta_1$. To see this, it is enough to show that $(X-\beta_1)(X+\beta_1)\in \mathbb{F}_q[X]$,
i.e., $\beta_1^2\in \mathbb{F}_q$. Noting that $\beta_1^{2q-2}=1$,
which implies $\beta_1^{2q}=\beta_1^{2}$, as claimed.

\begin{Theorem}\label{thm-cyclic}
With respect to  the above notations, we assume further that  $\gcd(\ell, q-1)=1$.
Let $C$ be a $\lambda$-constacyclic code
of length $2\ell^m p^n$ over $\mathbb F_q$. Then one of the following statements holds:

\item[{\bf(A)}]either $\lambda\in\langle \xi^2\rangle$, then
$a^{2\ell^m p^n}\lambda=1$ for some $a\in \mathbb F_q^*$, and we have
$$ C=\left\langle \prod\limits_{i=0}^{e}\hat M_{\rho_i}(aX)^{\varepsilon_i}\hat M_{\rho_i}(-aX)^{\epsilon_i}\right\rangle,
 \qquad 0\le \varepsilon_i, \epsilon_i\le p^n,~~\hbox{~for any ~$i=0,1,\cdots,e$};
$$

\item[{\bf(B)}]or $\lambda\notin\langle \xi^2\rangle$, then
$b^{2\ell^m p^n}\lambda=\xi^{p^n}$ for some $b\in \mathbb F_q^*$, and
there are two subcases:

\item[{\bf(B1)}]if $f={\rm ord}_{\ell}(q)$ is odd, we have that

$$ C=\left\langle
\prod\limits_{i=0}^{e}\hat S_i(bX)^{\varepsilon_i}\right\rangle,
 \qquad 0\le \varepsilon_i\le p^n, \hbox{~for any ~$i=0,1,\cdots,e,$}
$$
where  $S_i(X)=\hat M_{\rho_i}(\beta_1X)\hat M_{\rho_i}(-\beta_1X)$
 for each $0\leq i\leq e$.

\item[{\bf(B2)}]if $f={\rm ord}_{\ell}(q)$ is even, we have that

$$ C=\left\langle
\hat P(bX)^{\varepsilon}\prod\limits_{i=0}^{e}\hat Q_i(bX)^{\varepsilon_i}\hat R_i(bX)^{\epsilon_i}\right\rangle,
 \qquad 0\le \varepsilon, \varepsilon_i,\epsilon_i\le p^n, \hbox{~for any ~$i=0,1,\cdots,e,$}
$$
where  $P(X)=(X-\beta_1^{-1})(X+\beta_1^{-1})$, $Q_i(X)=\hat N_{\rho_i}(\beta_1X)\hat N_{\rho_iq}(-\beta_1X)$
and $R_{i}(X)=\hat N_{\rho_iq}(\beta_1X)\hat N_{\rho_i}(-\beta_1X)$  for each $0\leq i\leq e$.
\end{Theorem}
\begin{proof}
Since $\gcd(\ell,q-1)=1$, it is clear that $\langle\xi^{2\ell^mp^n}\rangle=\langle \xi^2\rangle$.
From $\lambda\in \langle \xi^2\rangle$, we have $\lambda^{-1}\in\langle \xi^2\rangle$, which implies that
an element $a\in \mathbb{F}_q^*$ can be found satisfying $a^{2\ell^m p^n}=\lambda^{-1}$, i.e., $a^{2\ell^m p^n}\lambda=1$.
By (\ref{cyclic-decomposition}),
$$
(aX)^{2\ell^mp^n}-1=\prod\limits_{i=0}^{e}M_{\rho_i}(aX)^{p^n}\hat M_{\rho_i}(-aX)^{p^n}.
$$
This leads to
$$
X^{2\ell^mp^n}-\lambda=X^{2\ell^mp^n}-a^{-2\ell^mp^n}=\prod\limits_{i=0}^{e}\hat M_{\rho_i}(aX)^{p^n}\hat M_{\rho_i}(-aX)^{p^n},
$$
proving $(A)$.

Assume now that $\lambda\notin\langle \xi^2\rangle$, which forces $\lambda\in\xi^{p^n}\langle \xi^2\rangle$.
We first give the irreducible factorization
of $X^{2\ell^mp^n}-\xi^{p^n}$ over $\mathbb F_q$.
Clearly, it suffices to determine the irreducible factors of $X^{2\ell^m}-\xi$ over $\mathbb F_q$.
As discussed previously, we take $\alpha_1$ to be an element in $\mathbb F_{q^2}$ so that
$\alpha_1^2=\xi$. That is,  we have the irreducible factorization of
$X^2-\xi$ in $\mathbb F_{q^2}[X]$,
$X^2-\xi=(X-\alpha_1)(X+\alpha_1)$. It follows that  $X^{2\ell^m}-\xi=(X^{\ell^m}-\alpha_1)(X^{\ell^m}+\alpha_1)$.
There exists an element $\beta_1\in \mathbb F_{q^2}$ such that $\beta_1^{\ell^m}\alpha_1=1$; furthermore,  $\beta_1$
satisfies $\beta_1^q=-\beta_1$.

At this point, the cases $f$ being odd and  even diverge.

Assume first that $f$ is odd. By the discussion in Section 2, we know that
$C_{\rho_0}=\{0\}$ and $C_{\rho_k}$, $1\leq k\leq e$,  consist all the distinct
$q^2$-cyclotomic cosets modulo $\ell^m$.
That is to say, Formula~(\ref{simple-irr-decomposition}) gives the irreducible factorization of $X^{\ell^m}-1$ over $\mathbb F_{q^2}$.
Substituting $\beta_1X$ for $X$ into Formula~(\ref{simple-irr-decomposition}),
we get the irreducible factorization of $X^{\ell^m}-\alpha_1$ over $\mathbb F_{q^2}$:
$$
X^{\ell^m}-\alpha_1=\hat M_{\rho_0}(\beta_1X)\hat M_{\rho_1}(\beta_1X)
 \cdots\hat M_{\rho_e}(\beta_1X).
$$
Similarly, we have the irreducible factorization of $X^{\ell^m}+\alpha_1$ over $\mathbb F_{q^2}$:
$$
X^{\ell^m}+\alpha_1=\hat M_{\rho_0}(-\beta_1X)\hat M_{\rho_1}(-\beta_1X)
 \cdots\hat M_{\rho_e}(-\beta_1X).
$$
Combining these results, we have
$$
X^{2\ell^m}-\xi=(X^{\ell^m}-\alpha_1)(X^{\ell^m}+\alpha_1)=
\prod\limits_{i=0}^{e}\hat M_{\rho_i}(\beta_1X)\hat M_{\rho_i}(-\beta_1X),
$$
which is the monic irreducible factorization of $X^{2\ell^m}-\xi$ over $\mathbb F_{q^2}$.
Let $S_i(X)=\hat M_{\rho_i}(\beta_1X)\hat M_{\rho_i}(-\beta_1X)$ for each $0\leq i\leq e$.
We claim that $S_i(X)$ is an irreducible polynomial over $\mathbb F_q$.
Recall that
$$
\hat M_{\rho_i}(\beta_1X)=\prod\limits_{k\in C_{\rho_i}}(X-\beta_1^{-1}\eta^k) ~~\hbox{and}~~
\hat M_{\rho_i}(-\beta_1X)=\prod\limits_{k\in C_{\rho_i}}(X+\beta_1^{-1}\eta^k).
$$
Then $\beta_1^{-1}\eta^k$ gives all the roots of $\hat M_{\rho_i}(\beta_1X)$ when $k$ ranges over $C_{\rho_i}$.
Now
$
(\beta_1^{-1}\eta^k)^q=\beta_1^{-q}\eta^{kq}=-\beta_1^{-1}\eta^{kq},
$
which is the root of $\hat M_{\rho_i}(-\beta_1X)$.
We deduce that $S_i(X)=\hat M_{\rho_i}(\beta_1X)\hat M_{\rho_i}(-\beta_1X)$
is an irreducible polynomial over $\mathbb F_q$, as claimed.
We get the irreducible factorization of $X^{2\ell^mp^n}-\xi$ over $\mathbb F_q$ as follows:
\begin{equation}\label{xi}
X^{2\ell^mp^n}-\xi^{p^n}=(X^{2\ell^m}-\xi)^{p^n}=
\prod\limits_{i=0}^{e}\big(\hat M_{\rho_i}(\beta_1X)\hat M_{\rho_i}(-\beta_1X)\big)^{p^n}
=\prod\limits_{i=0}^{e}S_i(X)^{p^n}.
\end{equation}
Since $\lambda\in\xi^{p^n}\langle \xi^{2\ell^mp^n}\rangle$,
there exists an element $b\in \mathbb F_q^*$ such that $b^{2\ell^m p^n}\lambda=\xi^{p^n}$.
We establish  the following $\mathbb F_q$-algebra isomorphism:
\begin{equation*}
\quad \mathbb F_q[X]/\langle X^{2\ell^mp^n}-\xi^{p^n}\rangle~\longrightarrow~
  \mathbb F_q[X]/\langle X^{2\ell^mp^n}-\lambda\rangle,~~f(X)\longrightarrow f(bX).
\end{equation*}
By (\ref{xi}), we get the irreducible factorization of $X^{2\ell^mp^n}-\lambda$ over $\mathbb F_q$:
$$
X^{2\ell^mp^n}-\lambda=\prod\limits_{i=0}^{e}\hat S_i(bX)^{p^n},
$$
which gives the desired result.

It remains to consider the case when $f$ is even. It is known that $D_0$, $D_{\ell^{m-r}g^j}$ and
$D_{\ell^{m-r}g^jq}$ consist all the distinct $q^2$-cyclotomic cosets modulo $\ell^m$,
where $0\leq j\leq\delta(r)-1$ and $1\leq r\leq m$. That is,
\begin{equation*}
X^{\ell^m}-1=(X-1) N_{\rho_1}(X)N_{\rho_1q}(X)N_{\rho_2}(X)N_{\rho_2q}(X)
 \cdots N_{\rho_e}(X)N_{\rho_eq}(X),
\end{equation*}
gives the irreducible factorization of $X^{\ell^m}-1$ over $\mathbb F_{q^2}$ as has been shown in (\ref{simple-irr-decomposition2}).
Working as the same with the case of $f$ being odd, we get the irreducible factorizations of $X^{\ell^m}-\alpha_1$
and  $X^{\ell^m}+\alpha_1$ over $\mathbb F_{q^2}$:
$$
X^{\ell^m}-\alpha_1=(X-\beta_1^{-1})\hat N_{\rho_1}(\beta_1X)\hat N_{\rho_1q}(\beta_1X)\hat N_{\rho_2}(\beta_1X)\hat N_{\rho_2q}(\beta_1X)
 \cdots \hat N_{\rho_e}(\beta_1X)\hat N_{\rho_eq}(\beta_1X),
$$
$$
X^{\ell^m}+\alpha_1=(X+\beta_1^{-1})\hat N_{\rho_1}(-\beta_1X)\hat N_{\rho_1q}(-\beta_1X)\hat N_{\rho_2}(-\beta_1X)\hat N_{\rho_2q}(-\beta_1X)
 \cdots \hat N_{\rho_e}(-\beta_1X)\hat N_{\rho_eq}(-\beta_1X).
$$
Now the irreducible factorization of $X^{2\ell^m}-\xi$ over $\mathbb F_{q^2}$ is given by
$$
X^{2\ell^m}-\xi=(X^{\ell^m}-\alpha_1)(X^{\ell^m}+\alpha_1)=(X-\beta_1^{-1})(X+\beta_1^{-1})
\prod\limits_{i=0}^{e}\hat N_{\rho_i}(\beta_1X)\hat N_{\rho_iq}(\beta_1X)\hat N_{\rho_i}(-\beta_1X)\hat N_{\rho_iq}(-\beta_1X).
$$
Let $P(X)=(X-\beta_1^{-1})(X+\beta_1^{-1})$, $Q_i(X)=\hat N_{\rho_i}(\beta_1X)\hat N_{\rho_iq}(-\beta_1X)$
and $R_{i}(X)=\hat N_{\rho_iq}(\beta_1X)\hat N_{\rho_i}(-\beta_1X)$  for each $0\leq i\leq e$.
Using arguments similar to the proof in (A), we see that $P(X)$, $Q_i(X)$ and $R_i(X)$ are irreducible polynomials over $\mathbb F_q$.
Then we get the irreducible factorization of $X^{2\ell^mp^n}-\xi^{p^n}$ over $\mathbb F_q$ as follows:
\begin{equation}\label{xi2}
X^{2\ell^mp^n}-\xi^{p^n}=(X^{2\ell^m}-\xi)^{p^n}=P(X)^{p^n}
\prod\limits_{i=0}^{e}Q_i(X)^{p^n}R_i(X)^{p^n}.
\end{equation}
Finally, we get the irreducible factorization of $X^{2\ell^mp^n}-\lambda$ over $\mathbb F_q$:
$$
X^{2\ell^mp^n}-\lambda=\hat P(bX)^{p^n}\prod\limits_{i=0}^{e}\hat Q_i(bX)^{p^n}\hat R_i(bX)^{p^n}.
$$
\end{proof}

Next we consider the  case $\gcd(\ell, q-1)=\ell$, namely $\ell\mid(q-1)$.
We use Lemma~\ref{n-equivalence} again to obtain the concerning results.
We first adopt the following notations:
$
\ell^u\Vert(q-1),v=\min\{m,u\}$ and $\zeta=\xi^{\frac{q-1}{\ell^v}}.
$

\begin{Theorem}\label{thm5} With respect to the above notations, we assume further that $\ell\mid(q-1)$.
For any nonzero element $\lambda$ of $\mathbb F_q$ and any
$\lambda$-constacyclic code $C$ of length $2\ell^mp^n$ over $\mathbb F_q$,
one of the following  holds:

\item[{\bf (I)}]$\lambda\in\langle \xi^{2\ell^v}\rangle$, then $c_1^{2\ell^mp^n}\lambda=1$ for an element
$c_1\in \mathbb F_q$ and we have (The empty product is taken to be $1$):
\begin{equation*}
C=\left\langle\prod\limits_{i=0}^{\ell^v-1}\big(X-c_1^{-1}\zeta^i\big)^{\varepsilon_i}\big(X+c_1^{-1}\zeta^i\big)^{\epsilon_i}\cdot
\prod\limits_{j=1}^{m-u}\prod\limits_{k=1 \atop{ \ell\,\nmid \,k}}^{\ell^v}
\big(X^{\ell^j}-c_1^{-\ell^j}\zeta^k\big)^{\tau^j_k}\big(X^{\ell^j}+c_1^{-\ell^j}\zeta^k\big)^{\sigma_k^j}\right\rangle,
\end{equation*}
where $0\le  \varepsilon_i,\epsilon_i\le p^n$ for any $0\leq i\leq\ell^v-1$,  and
$0\leq\tau_k^j,\sigma_k^j\le p^n$ for each $1\leq j\leq m-u$ and $1\leq k\leq\ell^v$ with $\ell\nmid k.$

\item[{\bf (II)}]$\lambda\in\xi^{\ell^vp^n}\langle \xi^{2\ell^v}\rangle$, then $c_2^{2\ell^mp^n}\lambda=\xi^{\ell^vp^n}$ for an element
$c_2\in \mathbb F_q$ and one of the following  holds:

\item[{\bf(II.A)}]if $m\leq u$, then
$$
C=\left\langle\prod\limits_{i=0}^{\ell^m-1}(X^2-c_2^{-2}\xi\alpha^i)^{\varepsilon_i}\right\rangle,
 \qquad 0\le \varepsilon_i\le p^n, ~~\hbox{~for any ~$i=0,1,\cdots,\ell^m-1$},
$$
where $\alpha=\xi^{\frac{q-1}{\ell^m}}$ is a primitive $\ell^m$th root of unity in $\mathbb F_q;$

\item[{\bf(II.B)}]otherwise, we have that

\begin{equation*}
C=\left\langle\prod\limits_{i=0}^{\ell^u-1}\big(X^2-c_2^{-2}\beta^{-1}\zeta^i\big)^{\varepsilon_i}\cdot
\prod\limits_{j=1}^{m-u}\prod\limits_{k=1 \atop{ \ell\,\nmid \,k}}^{\ell^u}
 \big(X^{2\ell^j}-c_2^{-2\ell^j}\beta^{-\ell^j}\zeta^k\big)^{\sigma_k^j}\right\rangle,
 ~~~~\qquad 0\le  \varepsilon_i, \sigma_k^j\le p^n
\end{equation*}
where
$\beta$ is an element  in $\langle \xi^{\ell^u}\rangle$ such that $\beta^{\ell^m}\xi^{\ell^u}=1$.

\item[{\bf (III)}]$\lambda\in\xi^{jp^n}\langle \xi^{2\ell^v}\rangle$ with $1\leq j\leq 2\ell^v-1$ except $j=\ell^v$,
then there exists   $d_1\in \mathbb F_q^*$ such that $d_1^{2\ell^mp^n}\lambda=\xi^{jp^n};$
write $j=y\ell^z$ with $\gcd(y,\ell)=1$ and $0\le z\le v-1$.
There are two subcases:

\item[{\bf(III.A)}] if the integer $y$ is odd, then we have
$$ C=
\huge\left\langle\prod\limits_{i=0}^{\ell^z-1}
 (X^{2\ell^{m-z}}-d_1^{-2\ell^{m-z}}\delta^i\xi^{y})^{\varepsilon_i}
 \huge\right\rangle,
 \qquad 0\leq \varepsilon_i \leq p^n;
$$

\item[{\bf(III.B)}] otherwise, writing $y=2y_0$,  we have
$$ C=
\huge\left\langle\prod\limits_{i=0}^{\ell^z-1}
(X^{\ell^{m-z}}-d_1^{-\ell^{m-z}}\delta^i\xi^{y_0})^{\varepsilon_i}
(X^{\ell^{m-z}}+d_1^{-\ell^{m-z}}\delta^{i}\xi^{y_0})^{\epsilon_i}
 \huge\right\rangle,
 \qquad 0\leq \varepsilon_i, \epsilon_i \leq p^n,
$$
where $\delta=\xi^{(q-1)/\ell^z}$
is a primitive $\ell^z$th root of unity in $\mathbb F_q$.
\end{Theorem}

\begin{proof}
Consider the multiplicative group $\mathbb F_q^*=\langle\xi\rangle$
which is a cyclic group of order $q-1$ generated by $\xi$.
It is easy to check that
$\langle\xi^{2\ell^mp^n}\rangle
=\langle\xi^{2\ell^m}\rangle=\langle\xi^{2\ell^v}\rangle$
and the index $|\mathbb F_q^*:\langle\xi^{2\ell^v}\rangle|=2\ell^v$.
Thus the multiplicative group $\mathbb F_q^*$ is decomposed into
disjoint union of cosets over the subgroup $\langle\xi^{2\ell^v}\rangle$ as follows:
\begin{equation}\label{codets}
\mathbb F_q^*=\langle\xi\rangle=\langle\xi^{2\ell^v}\rangle\cup\xi^{p^n}\langle\xi^{2\ell^v}\rangle\cup\dots\cup\xi^{(2\ell^{v}-1)p^n}\langle\xi^{2\ell^v}\rangle.
\end{equation}
So the element $\lambda$ of $\mathbb F_q^*$ belongs to exactly one
of the cosets, i.e. there is a unique integer $j$ with $0\le j\le 2\ell^v-1$ such that $\lambda\in\xi^{jp^n}\langle\xi^{2\ell^v}\rangle$.
We get that
$\lambda$ is $2\ell^mp^n$-equivalent to $\xi^{jp^n}$.

{\bf Case (I)}: $j=0$, i.e.,  $\lambda$ and $1$ are $2\ell^mp^n$-equivalent.
We have an element $c_1\in \mathbb F_q^*$
such that $c_1^{2\ell^mp^n}\lambda=1$.
It needs to obtain the irreducible factorization of $X^{2\ell^mp^n}-1$ over $\mathbb F_q$.
Obviously,
$$
X^{2\ell^mp^n}-1=(X^{\ell^mp^n}-1)(X^{\ell^mp^n}+1)=(X^{\ell^m}-1)^{p^n}(X^{\ell^m}+1)^{p^n}.
$$
By \cite[Theorem 3.1]{Chen2}, we have the irreducible factorization of  $X^{\ell^m}-1$ over $\mathbb F_q$
as follows (The empty product is taken to be $1$):
$$
X^{\ell^m}-1=\prod\limits_{i=0}^{\ell^v-1}\big(X-\zeta^i\big)\cdot
\prod\limits_{j=1}^{m-u}\prod\limits_{k=1 \atop{ \ell\,\nmid \,k}}^{\ell^v}
 \big(X^{\ell^j}-\zeta^k\big).
$$
Since $\ell$ is odd, we can easily get the irreducible factorization of  $X^{\ell^m}+1$ over $\mathbb F_q$:
$$
X^{\ell^m}+1=\prod\limits_{i=0}^{\ell^v-1}\big(X+\zeta^i\big)\cdot
\prod\limits_{j=1}^{m-u}\prod\limits_{k=1 \atop{ \ell\,\nmid \,k}}^{\ell^v}
 \big(X^{\ell^j}+\zeta^k\big).
$$
Then
$$
X^{2\ell^mp^n}-1=\prod\limits_{i=0}^{\ell^v-1}\big(X-\zeta^i\big)^{p^n}\big(X+\zeta^i\big)^{p^n}\cdot
\prod\limits_{j=1}^{m-u}\prod\limits_{k=1 \atop{ \ell\,\nmid \,k}}^{\ell^v}
\big(X^{\ell^j}-\zeta^k\big)^{p^n}\big(X^{\ell^j}+\zeta^k\big)^{p^n}.
$$
Hence
$$
X^{2\ell^mp^n}-\lambda=\prod\limits_{i=0}^{\ell^v-1}\big(X-c_1^{-1}\zeta^i\big)^{p^n}\big(X+c_1^{-1}\zeta^i\big)^{p^n}\cdot
\prod\limits_{j=1}^{m-u}\prod\limits_{k=1 \atop{ \ell\,\nmid \,k}}^{\ell^v}
\big(X^{\ell^j}-c_1^{-\ell^j}\zeta^k\big)^{p^n}\big(X^{\ell^j}+c_1^{-\ell^j}\zeta^k\big)^{p^n}.
$$
The conclusion (I) holds.

{\bf Case (II)}: $j=\ell^v$.
We have an element $c_2\in \mathbb F_q^*$
such that $c_2^{2\ell^mp^n}\lambda=\xi^{\ell^vp^n}$.
We need to obtain the irreducible factorization of $X^{2\ell^m}-\xi^{\ell^v}$ over $\mathbb F_q$.
There are two subcases, namely $m\leq u$ and $m>u$.

If $m\leq u$ then $m=\min\{m,u\}=v$. We assume that
$\alpha=\xi^{\frac{q-1}{\ell^m}}$ is a primitive $\ell^m$th root of unity in $\mathbb F_q.$
Thus,
$$
X^{2\ell^mp^n}-\xi^{\ell^vp^n}=\big(X^{2\ell^m}-\xi^{\ell^m}\big)^{p^n}=\prod\limits_{i=0}^{\ell^m-1}(X^2-\xi\alpha^i)^{p^n},
$$
gives the irreducible factorization of $X^{2\ell^mp^n}-\xi^{\ell^vp^n}$ over $\mathbb F_q$ (Use \cite[Theorem 3.75]{Lidl}).

Otherwise,  $u=\min\{m,u\}=v$.
Then there exists an element $\beta$ in $\langle \xi^{\ell^u}\rangle$ such that $\beta^{\ell^m}\xi^{\ell^u}=1$. Indeed,
$\psi:~\langle \xi^{\ell^u}\rangle\rightarrow\langle \xi^{\ell^u}\rangle, ~x\mapsto x^{\ell^m},$ is a group automorphism.
This implies that a unique element $\beta\in\langle \xi^{\ell^u}\rangle$ can be found such that $\psi(\beta)=\beta^{\ell^m}=\xi^{-\ell^u}$,
i.e., $\beta^{\ell^m}\xi^{\ell^u}=1$.
In particular, $\beta$ is a primitive $\frac{q-1}{\ell^u}$th root of unity.
We get the irreducible factorization of $X^{\ell^m}-\xi^{\ell^u}$
as follows:
$$
X^{\ell^m}-\xi^{\ell^u}=\prod\limits_{i=0}^{\ell^u-1}\big(X-\beta^{-1}\zeta^i\big)\cdot
\prod\limits_{j=1}^{m-u}\prod\limits_{k=1 \atop{ \ell\,\nmid \,k}}^{\ell^u}
 \big(X^{\ell^j}-\beta^{-\ell^j}\zeta^k\big).
$$
Useing \cite[Theorem 3.75]{Lidl}, it is easily checked  that
$$
X^{2\ell^m}-\xi^{\ell^u}=\prod\limits_{i=0}^{\ell^u-1}\big(X^2-\beta^{-1}\zeta^i\big)\cdot
\prod\limits_{j=1}^{m-u}\prod\limits_{k=1 \atop{ \ell\,\nmid \,k}}^{\ell^u}
 \big(X^{2\ell^j}-\beta^{-\ell^j}\zeta^k\big)
$$
gives the irreducible factorization of $X^{2\ell^m}-\xi^{\ell^u}$ over $\mathbb F_q$.

{\bf Case (III)}: $0<j<2\ell^v$ except $j=\ell^v$.  We can assume that $j=y\ell^z$ with $\gcd(y,\ell)=1$ and $0\le z\le v-1$.
Since $\ell^{v}\,|\,(q-1)$, we see that $\delta=\xi^{(q-1)/\ell^z}$
is a primitive $\ell^z$th root of unity in~$\mathbb F_q$.
Noting that $z<m$, we have
$$\left(\frac{X^{\ell^{m-z}}}{\xi^{y}}\right)^{\ell^z}-1
 =\left(\frac{X^{\ell^{m-z}}}{\xi^{y}}-1\right)
 \left(\frac{X^{\ell^{m-z}}}{\xi^{y}}-\delta\right)\cdots\left(\frac{X^{\ell^{m-z}}}{\xi^{y}}-\delta^{\ell^z-1}\right),$$
hence
$$\left(\frac{X^{\ell^{m-z}}}{\xi^{y}}\right)^{\ell^zp^n}-1
 =\left(\frac{X^{\ell^{m-z}}}{\xi^{y}}-1\right)^{p^n}
 \left(\frac{X^{\ell^{m-z}}}{\xi^{y}}-\delta\right)^{p^n}\cdots\left(\frac{X^{\ell^{m-z}}}{\xi^{y}}-\delta^{\ell^z-1}\right)^{p^n};$$
that is,
$$
X^{\ell^mp^n}-\xi^{y\ell^zp^n}=(X^{\ell^{m-z}}-\xi^{y})^{p^n}
 (X^{\ell^{m-z}}-\delta\xi^{y})^{p^n}
  \cdots(X^{\ell^{m-z}}-\delta^{\ell^z-1}\xi^{y})^{p^n}.
$$
Thus
\begin{equation}\label{irreducible1}
X^{2\ell^mp^n}-\xi^{y\ell^zp^n}=(X^{2\ell^{m-z}}-\xi^{y})^{p^n}
 (X^{2\ell^{m-z}}-\delta\xi^{y})^{p^n}
  \cdots(X^{2\ell^{m-z}}-\delta^{\ell^z-1}\xi^{y})^{p^n}.
\end{equation}
Now we need to give the irreducible factorization of
$X^{2\ell^mp^n}-\xi^{y\ell^zp^n}$ over $\mathbb F_q$. There are two subcases:

\item[{\bf(III.A)}]
The integer $y$ is odd. In this case, we assert that
Equation~(\ref{irreducible1}) gives the irreducible
factorization of
$X^{2\ell^mp^n}-\xi^{y\ell^zp^n}$ over $\mathbb F_q$.
It suffices to check that each polynomial $X^{2\ell^{m-z}}-\delta^i\xi^{y}$
is irreducible over $\mathbb F_q$, $0\leq i\leq\ell^z-1$.
Recall that $\ell\nmid y$,  $\ell^z<\ell^m$ and $\ell^u\Vert{\rm ord}(\xi)$. One can check that
$\ell\mid{\rm ord}(\delta^i\xi^{y})$ and $\ell\nmid\frac{q-1}{{\rm ord}(\delta^i\xi^{y})};$
meanwhile $2\mid{\rm ord}(\delta^i\xi^{y})$ and $2\nmid\frac{q-1}{{\rm ord}(\delta^i\xi^{y})}$.
Using \cite[Theorem 3.75]{Lidl}, we get the desired result.

\item[{\bf(III.B)}]
We are left to investigate the case  $y=2y_0$. Clearly, $X^{2\ell^m}-\xi^{2y_0\ell^z}=(X^{\ell^m}-\xi^{y_0\ell^z})(X^{\ell^m}+\xi^{y_0\ell^z})$.
Hence, we get the irreducible factorization of $X^{2\ell^mp^n}-\xi^{y\ell^zp^n}$ over $\mathbb F_q$:
\begin{equation*}
X^{2\ell^mp^n}-\xi^{y\ell^zp^n}=\prod\limits_{i=0}^{\ell^z-1}
(X^{\ell^{m-z}}-\delta^i\xi^{y_0})^{p^n}
(X^{\ell^{m-z}}+\delta^{i}\xi^{y_0})^{p^n}.
\end{equation*}

\end{proof}

\section{Dual codes}
In this section, the duals of all constacyclic codes of length $2\ell^mp^n$ over $\mathbb F_{q}$ are explicitly obtained,
where $\ell$ is an odd prime different from $p$.
Among other results, all  linear complementary-dual (LCD) cyclic and
negacylic codes  are provided; all self-dual negacyclic codes of this length are also determined.

We give our results according to Theorem~\ref{thm-cyclic} and Theorem~\ref{thm5}.
The next two results give the structures of the duals of all constacyclic codes
of length $2\ell^mp^n$ over $\mathbb F_q$.

\begin{Corollary}\label{forLCD1}
With the notation of Theorem~\ref{thm-cyclic}, we have that

\item[{\bf (A)}]   $\lambda\in\langle \xi^2\rangle$.
If  $C$ is a $\lambda$-constacyclic code
presenting in Theorem~\ref{thm-cyclic}~{\rm (A)},
then its dual code is the $\lambda^{-1}$-constacyclic code given by
$$
C^\perp=\left\langle \prod\limits_{i=0}^{e}\hat M_{-\rho_i}(a^{-1}X)^{p^n-\varepsilon_i}\hat M_{-\rho_i}(-a^{-1}X)^{p^n-\epsilon_i}\right\rangle;
$$

\item[{\bf (B1)}]   $\lambda\notin\langle \xi^2\rangle$ and $f$ is odd.
If  $C$ is a $\lambda$-constacyclic code
presenting in Theorem~\ref{thm-cyclic}~{\rm (B1)},
then its dual code is the $\lambda^{-1}$-constacyclic code given by
$$
C^\perp=\left\langle \prod\limits_{i=0}^{e}\hat S_{-i}(b^{-1}X)^{p^n-\varepsilon_i}\right\rangle,
$$
where  $S_{-i}(X)=\hat M_{-\rho_i}(\beta_1^{-1}X)\hat M_{-\rho_i}(-\beta_1^{-1}X)$
 for each $0\leq i\leq e$;

\item[{\bf (B2)}]   $\lambda\notin\langle \xi^2\rangle$ and $f$ is even.
If  $C$ is a $\lambda$-constacyclic code
presenting in Theorem~\ref{thm-cyclic}~{\rm (B2)},
then its dual code is the $\lambda^{-1}$-constacyclic code given by
$$
C^\perp=\left\langle
\hat P^*(b^{-1}X)^{p^n-\varepsilon}\prod\limits_{i=0}^{e}\hat Q_{-i}(b^{-1}X)^{p^n-\varepsilon_i}\hat R_{-i}(b^{-1}X)^{p^n-\epsilon_i}
\right\rangle,
$$
\small{
where  $P^*(X)=(X-\beta_1)(X+\beta_1)$, $Q_{-i}(X)=\hat N_{-\rho_i}(\beta_1^{-1}X)\hat N_{-\rho_iq}(-\beta_1^{-1}X)$
and
$R_{-i}(X)=\hat N_{-\rho_iq}(\beta_1^{-1}X)\hat N_{-\rho_i}(-\beta_1^{-1}X)$  for each $0\leq i\leq e$.}
\end{Corollary}
\begin{proof}
We just give a proof for (A), the other cases can be proved similarly.
As shown in the proof of Theorem~\ref{thm-cyclic}, the monic irreducible factorization of $X^{2\ell^mp^n}-\lambda$ over
$\mathbb{F}_q$ is given by
$$
X^{2\ell^mp^n}-\lambda=\prod\limits_{i=0}^{e}\hat M_{\rho_i}(aX)^{p^n}\hat M_{\rho_i}(-aX)^{p^n}.
$$
Then
$$
h(X)=\frac{X^{2\ell^mp^n}-\lambda}{g(X)}=\prod\limits_{i=0}^{e}\hat M_{\rho_i}(aX)^{p^n-\varepsilon_i}\hat M_{\rho_i}(-aX)^{p^n-\epsilon_i}.
$$
It follows that $C^{\perp}$ has generator polynomial
$$
h^*(X)=\prod\limits_{i=0}^{e}\hat M_{-\rho_i}(a^{-1}X)^{p^n-\varepsilon_i}\hat M_{-\rho_i}(-a^{-1}X)^{p^n-\epsilon_i}.
$$
\end{proof}

The next result is a direct consequence of
Theorem \ref{thm5}, so we omit its proof here.

\begin{Corollary}
With the notation of Theorem \ref{thm5}, we have that

\item[{\bf (I)}]~If $C$ is a $\lambda$-constacyclic code
given as in Theorem~\ref{thm5}~{\rm (I)},
then its dual code is the $\lambda^{-1}$-constacyclic code given by
\begin{equation*}
C^\perp=\left\langle\prod\limits_{i=0}^{\ell^v-1}\big(X-c_1\zeta^{-i}\big)^{p^n-\varepsilon_i}\big(X+c_1\zeta^{-i}\big)^{p^n-\epsilon_i}\cdot
\prod\limits_{j=1}^{m-u}\prod\limits_{k=1 \atop{ \ell\,\nmid \,k}}^{\ell^v}
\big(X^{\ell^j}-c_1^{\ell^j}\zeta^{-k}\big)^{p^n-\tau^j_k}\big(X^{\ell^j}+c_1^{\ell^j}\zeta^{-k}\big)^{p^n-\sigma_k^j}\right\rangle.
\end{equation*}

\item[{\bf (II.A)}] If $C$ is a $\lambda$-constacyclic code
given as in Theorem~\ref{thm5}~{\rm (II.A)},
then its dual code is the $\lambda^{-1}$-constacyclic code given by
$$ C^\perp=
\huge\left\langle\prod\limits_{i=0}^{\ell^m-1}(X^2-c_2^{2}\xi^{-1}\alpha^{-i})^{p^n-\varepsilon_i}
 \huge\right\rangle.
$$

\item[{\bf (II.B)}] If $C$ is a $\lambda$-constacyclic code
given as in Theorem~\ref{thm5}~{\rm (II.B)},
then its dual code  is the $\lambda^{-1}$-constacyclic code given by
$$ C^\perp=
\huge\left\langle\prod\limits_{i=0}^{\ell^u-1}\big(X^2-c_2^{2}\beta\zeta^{-i}\big)^{p^n-\varepsilon_i}\cdot
\prod\limits_{j=1}^{m-u}\prod\limits_{k=1 \atop{ \ell\,\nmid \,k}}^{\ell^u}
 \big(X^{2\ell^j}-c_2^{2\ell^j}\beta^{\ell^j}\zeta^{-k}\big)^{p^n-\sigma_k^j}
 \huge\right\rangle.
$$

\item[{\bf (III.A)}] If $C$ is a $\lambda$-constacyclic code
given  as in Theorem~\ref{thm5}~{\rm (III.A)},
then its dual code  is the $\lambda^{-1}$-constacyclic code given by
$$ C^\perp=
\huge\left\langle\prod\limits_{i=0}^{\ell^z-1}
 (X^{2\ell^{m-z}}-d_1^{2\ell^{m-z}}\delta^{-i}\xi^{-y})^{p^n-\varepsilon_i}
 \huge\right\rangle.
$$

\item[{\bf (III.B)}] If $C$ is a $\lambda$-constacyclic code
given as in Theorem~\ref{thm5}~{\rm (III.B)},
then its dual code is the $\lambda^{-1}$-constacyclic code given by
$$ C^\perp=
\huge\left\langle\prod\limits_{i=0}^{\ell^z-1}
(X^{\ell^{m-z}}-d_1^{\ell^{m-z}}\delta^{-i}\xi^{-y_0})^{p^n-\varepsilon_i}
(X^{\ell^{m-z}}+d_1^{\ell^{m-z}}\delta^{-i}\xi^{-y_0})^{p^n-\epsilon_i}
 \huge\right\rangle.
$$

\end{Corollary}

We devote the rest of this section to apply our results  to investigate the situations of  linear complimentary-dual (LCD) codes
and self-dual codes. These are the two extreme connections between $C$ and $C^{\perp}$, where $C\bigcap C^{\perp}=\{0\}$ (for LCD codes) and
$C=C^{\perp}$ (for self-dual codes). The concept of LCD codes was introduced by Massey \cite{M92} in 1992. In the same paper, he showed that asymptotically good
LCD codes exist, and presented applications of LCD codes such as they provided  an optimum linear coding
solution for the two-user binary adder channel. It was proven by Sendrier \cite{S04} that LCD codes meet the Gilbert-Varshamov bound. Necessary and sufficient conditions for cyclic codes \cite{YM94} and certain class of quasi-cyclic codes \cite{EY09} to be LCD codes were obtained.

For the case of LCD constacyclic codes, it was shown that any $\lambda$-constacyclic
code with $\lambda \not\in \{-1,1\}$ is a LCD code (\cite{Dinh2014}).
So in order to obtain all LCD $\lambda$-constacyclic codes, we only need to work on cyclic and negacyclic codes.

Recall that ${\rm ord}_{\ell}(q)=f$, the multiplicative order of $q$ in $\mathbb{Z}_{\ell}^*$.
Also recall that ${\rm ord}_{\ell^{r}}(q)=\lambda(r)$ and $\delta(r)=\frac{\phi(\ell^{r})}{\lambda(r)}$,
$1\leq r\leq m$.
We have to distinguish the cases when $f$ is  odd or even.
If $f={\rm ord}_{\ell}(q)$ is even, it has been shown that
the monic irreducible factors of $X^{\ell^m}-1$ are self-reciprocal (e.g. see \cite[Theorem 1]{Kathuria}).
The next lemma is concerned with the case  when $f$ is odd, in which case it shows that all the irreducible factors
of $X^{\ell^m}-1$ are not self-reciprocal except the trivial factor $X-1$.

\begin{lem}\label{lem2}
Let $\ell$ be an odd prime relatively prime to  $q$.  Assume further that $f={\rm ord}_{\ell}(q)$ is odd.
If $g$ is a fixed generator of the cyclic group $\mathbb{Z}_{\ell^m}^*$, then
all the distinct $q$-cyclictomic cosets modulo $\ell^m$ are given by
$C_0=\{0\}$,
$$
C_{\ell^{m-r}g^k}=\Big\{\ell^{m-r}g^k,\ell^{m-r}g^kq,\cdots,\ell^{m-r}g^kq^{\lambda(r)-1}\Big\},
$$
$$
C_{-\ell^{m-r}g^k}=\Big\{-\ell^{m-r}g^k, -\ell^{m-r}g^kq,\cdots,-\ell^{m-r}g^kq^{\lambda(r)-1}\Big\},
$$
where $1\leq r\leq m$, $0\leq k\leq\frac{\delta(r)}{2}-1$.
\end{lem}
\begin{proof}
We first claim that the cyclotomic cosets $C_{\ell^{m-r}g^k}$ and $C_{-\ell^{m-r}g^k}$
with $1\leq r\leq m$, $0\leq k\leq\frac{\delta(r)}{2}-1$ are distinct from each other.
Suppose otherwise that $C_{z\ell^{m-r_1}g^{k_1}}=C_{\ell^{m-r_2}g^{k_2}}$ for some
$1\leq r_1, r_2\leq m$, $0\leq k_1, k_2\leq\frac{\delta(r)}{2}-1$ and $z\in\{1,-1\}$.
Then there exists some integer $j$ such that
\begin{equation}\label{4-1}
z\ell^{m-r_1}g^{k_1}\equiv\ell^{m-r_2}g^{k_2}q^j~(\bmod~\ell^m).
\end{equation}
It follows that $\gcd(z\ell^{m-r_1}g^{k_1}, \ell^m)=\gcd(\ell^{m-r_2}g^{k_2}q^j,\ell^m)$, which forces $r_1=r_2$.
Therefore,
$
zg^{k_1}\equiv g^{k_2}q^j~(\bmod~\ell^{r_1}),
$
which gives
$
g^{2\lambda(r_1)k_1}\equiv g^{2\lambda(r_1)k_2}~(\bmod~\ell^{r_1}).
$
We get $2\lambda(r_1)(k_1-k_2)\equiv0~(\bmod~\phi(\ell^{r_1}))$,  since $g$ is of
order $\phi(\ell^{r_1})$ in $\mathbb{Z}_{\ell^{r_1}}^*$.  Hence, $k_1=k_2$.
Then (\ref{4-1}) gives
$
z\equiv q^{j}~(\bmod~\ell^{r_1}).
$
Using ${\rm ord}_{\ell^{r_1}}(q)=\lambda(r_1)$ again, we have that $\lambda(r_1)$ divides $2j$.
Since $\lambda(r_1)$ is odd ($f$ and $\ell$ are all the divisors of $\lambda(r_1)$), we get $\lambda(r_1)$ divides $j$. This leads to
$
1\equiv q^{j}~(\bmod~\ell^{r_1}).
$
We then see that $z=1$, as  claimed.

Finally,
$$
|C_0|+\sum_{r=1}^{m}\sum_{k=0}^{\frac{\delta(r)}{2}-1}(|C_{\ell^{m-r}g^k}|+|C_{-\ell^{m-r}g^k}|)
=1+\sum_{r=1}^{m}\sum_{k=0}^{\frac{\delta(r)}{2}-1}2\lambda(r)=1+\sum_{r=1}^{m}\lambda(r)\delta(r)=\ell^m.
$$
This completes the proof.
\end{proof}

Assuming that ${\rm ord}_{\ell}(q)=f$ is odd, by Lemma~\ref{lem2},
\begin{equation}\label{another-factorization}
X^{\ell^m}-1=(X-1)\prod\limits_{r=1}^{m}\prod\limits_{k=0}^{\frac{\delta(r)}{2}-1}
M_{\ell^{m-r}g^k}(X)M_{-\ell^{m-r}g^k}(X)
\end{equation}
gives the irreducible factorization of $X^{\ell^m}-1$ over $\mathbb F_q$.
Clearly, in this case, $(X-1)^*=X-1$ and $M_{\ell^{m-r}g^k}^*(X)=M_{-\ell^{m-r}g^k}(X)$ for each
$1\leq r\leq m$, $0\leq k\leq\frac{\delta(r)}{2}-1$.

The next result characterizes all LCD cyclic codes of length $2\ell^mp^n$ over $\mathbb{F}_q$.

\begin{Theorem}\label{LCDCYCLIC}
Let $\ell$ be an odd prime different from $p$. The following statements hold:

\item[{(i)}]If $f={\rm ord}_{\ell}(q)$ is odd, then there are exactly $2^{e+2}$ LCD cyclic codes of length $2\ell^mp^n$
over $\mathbb F_q$ generated by
$$(X-1)^{\varepsilon_0}(X+1)^{\epsilon_0}
\prod\limits_{r=1}^{m}\prod\limits_{k=0}^{\frac{\delta(r)}{2}-1}
M_{\ell^{m-r}g^k}(X)^{\varepsilon^r_k}M_{-\ell^{m-r}g^k}(X)^{\varepsilon^r_k}
\hat M_{\ell^{m-r}g^k}(-X)^{\sigma^r_k}\hat M_{-\ell^{m-r}g^k}(-X)^{\sigma^r_k},
$$
where $\varepsilon_0, \epsilon_0, \varepsilon^r_k, \sigma^r_k\in \{0, p^n\} $,
for every $1\leq r\leq m$
and $0\leq k\leq \frac{\delta(r)}{2}-1$;

\item[{(ii)}]if $f={\rm ord}_{\ell}(q)$ is even, then there are exactly $2^{2(e+1)}$ LCD cyclic codes of length $2\ell^mp^n$
over $\mathbb F_q$  generated by
$$\prod\limits_{i=0}^{e}\hat M_{\rho_i}(X)^{\varepsilon_i}\hat M_{\rho_i}(-X)^{\epsilon_i},
 \qquad \varepsilon_i, \epsilon_i\in \{0, p^n\},~~i=0,1,\cdots,e.
$$
\end{Theorem}
\begin{proof}
We just give a proof for ${\rm (i)}$, since the proof for ${\rm (ii)}$ is similar.
We get the desired result by computing the intersection of $C$
and $C^\perp$.  Form Lemma~\ref{lem2} and (\ref{another-factorization}),
we can assume that $C$ is a cyclic code of length $2\ell^mp^n$ over $\mathbb F_q$
generated by
$$(X-1)^{\varepsilon_0}(X+1)^{\epsilon_0}
\prod\limits_{r=1}^{m}\prod\limits_{k=0}^{\frac{\delta(r)}{2}-1}
M_{\ell^{m-r}g^k}(X)^{\varepsilon^r_k}M_{-\ell^{m-r}g^k}(X)^{\epsilon^r_k}
\hat M_{\ell^{m-r}g^k}(-X)^{\sigma^r_k}\hat M_{-\ell^{m-r}g^k}(-X)^{\tau^r_k},
$$
where $0\le \varepsilon_0, \epsilon_0\le p^n$,
$0\leq\varepsilon^r_k, \epsilon^r_k, \sigma^r_k, \tau^r_k\le p^n$ for every $1\leq r\leq m$
and $0\leq k\leq \frac{\delta(r)}{2}-1$.
Then its dual code $C^\perp$ has generator polynomial
\small{
\begin{equation*}
(X-1)^{p^n-\varepsilon_0}(X+1)^{p^n-\epsilon_0}\prod\limits_{r=1}^{m}\prod\limits_{k=0}^{\frac{\delta(r)}{2}-1}
M_{-\ell^{m-r}g^k}(X)^{p^n-\varepsilon^r_k}M_{\ell^{m-r}g^k}(X)^{p^n-\epsilon^r_k}
\hat M_{-\ell^{m-r}g^k}(-X)^{p^n-\sigma^r_k}\hat M_{\ell^{m-r}g^k}(-X)^{p^n-\tau^r_k}.
\end{equation*}}
\normalsize
Thus, $C\bigcap C^\perp=\{0\}$ if and only if
\begin{equation*}
\begin{split}
p^n&=\max\{\varepsilon_0, p^n-\varepsilon_0\}=\max\{\epsilon_0, p^n-\epsilon_0\}\\
&=\max \{\varepsilon_{k}^r,p^n-\epsilon^r_k\}=
\max \{\epsilon_{k}^r,p^n-\varepsilon^r_k\}\\
&=\max \{\sigma_{k}^r,p^n-\tau^r_k\}=\max \{\tau_{k}^r,p^n-\sigma^r_k\},
\end{split}
\end{equation*}
for every $1\leq r\leq m$
and $0\leq k\leq \frac{\delta(r)}{2}-1$,
which is equivalent to
$$
\varepsilon_0, \epsilon_0\in\{0, p^n\}, ~~~~\varepsilon_{k}^r,\sigma_{k}^r\in\{0, p^n\},~~~
\varepsilon_{k}^r=\epsilon_{k}^r,~~~ \sigma_{k}^r=\tau_{k}^r.
$$
We complete the proof of statement {\rm (i)}.
\end{proof}

Next we  give all  LCD negacyclic codes of length $2\ell^mp^n$ over $\mathbb F_{q}$.
Note that $1$ and $-1$ are $2\ell^mp^n$-equivalent  if $q\equiv1~(\bmod~4)$;
in this case,
let $\gamma$ be a primitive fourth root of unity in $\mathbb F_q$. That is, $X^2+1=(X-\gamma)(X+\gamma)$.
We  take an element $\beta$ in $\mathbb F_q$ so that $\beta^{\ell^m}\gamma=1$. Clearly, $\beta^{-1}=-\beta$.
If $q\equiv3~(\bmod~4)$, then $X^2+1$ is irreducible over $\mathbb F_q;$ let $\varsigma$ be an element in
$\mathbb F_{q^2}$ satisfying $X^2+1=(X-\varsigma)(X+\varsigma)$. We take $\theta$ in $\mathbb F_{q^2}$ so that
$\theta^{\ell^m}\varsigma=1$. It follows that $\theta^{-1}=\theta^q=-\theta$.

\begin{Theorem}\label{LCDNEGACYCLIC}Let $f={\rm ord}_{\ell}(q)$.  With the notation given above, we have that

{\rm (i)}~if $q\equiv1~(\bmod~4)$, then there are exactly $2^{e+1}$ LCD negacyclic codes of length $2\ell^mp^n$
over $\mathbb F_q$ generated by
$$
\prod\limits_{i=0}^{e}\hat M_{\rho_i}(\beta X)^{\varepsilon_i}\hat M_{-\rho_i}(-\beta X)^{\varepsilon_i},
~~~\varepsilon_i\in\{0, p^n\},~~~i=0,1,\cdots,e;
$$

{\rm (ii)}~if $q\equiv3~(\bmod~4)$ and $f$ is odd,
then there are exactly $2^{1+\frac{e}{2}}$ LCD negacyclic codes of length $2\ell^mp^n$
over $\mathbb F_q$ generated by
$$
(X^2+1)^{\varepsilon_0}\prod\limits_{r=1}^{m}\prod\limits_{k=0}^{\frac{\delta(r)}{2}-1}
I_{r,k}(X)^{\varepsilon^r_k}J_{r,k}(X)^{\varepsilon^r_k},
 \qquad \varepsilon_0, \varepsilon^r_k, \epsilon^r_k\in \{0, p^n\},
$$
where $I_{r,k}(X)=\hat M_{\ell^{m-r}g^k}(\theta X)\hat M_{\ell^{m-r}g^k}(-\theta X)$
and $J_{r,k}(X)=\hat M_{-\ell^{m-r}g^k}(\theta X)\hat M_{-\ell^{m-r}g^k}(-\theta X)$
for every $1\leq r\leq m$
and $0\leq k\leq \frac{\delta(r)}{2}-1$;

{\rm (iii)}~if $q\equiv3~(\bmod~4)$ and $f\equiv2~(\bmod~4)$, then there are exactly $2^{1+2e}$
LCD negacyclic codes of length $2\ell^mp^n$ over $\mathbb F_{q}$ generated by
$$
(X^2+1)^{\varepsilon_0}\prod\limits_{k=1}^{e}S_k(X)^{\varepsilon_k}T_k(X)^{\epsilon_k},
~~~ \qquad \varepsilon_0, \varepsilon_k, \epsilon_k\in \{0, p^n\},~~k=1,\cdots,e,
$$
where $S_k(X)=\hat N_{\rho_k}(\theta X)\hat N_{\rho_kq}(-\theta X)$
and $T_k(X)=\hat N_{\rho_k}(-\theta X)\hat N_{\rho_kq}(\theta X)$;

{\rm (iv)}~if $q\equiv3~(\bmod~4)$ and $f\equiv0~(\bmod~4)$, then there are exactly $2^{1+e}$
LCD negacyclic codes of length $2\ell^mp^n$ over $\mathbb F_q$
generated by
$$
(X^2+1)^{\varepsilon_0}\prod\limits_{k=1}^{e}S_k(X)^{\varepsilon_k}T_k(X)^{\varepsilon_k},
~~~ \qquad \varepsilon_0, \varepsilon_k\in \{0, p^n\},~~k=1,\cdots,e.
$$
\end{Theorem}
\begin{proof}
{\rm (i)}~We first indicate that,
$C_{\rho_0}=\{0\}$ and
\begin{equation*}\label{another}
C_{-\ell^{m-r}g^k}=\{-\ell^{m-r}g^k, -\ell^{m-r}g^kq,\cdots,-\ell^{m-r}g^kq^{\lambda(r)-1}\},
~0\leq k\leq\delta(r)-1, ~1\leq r\leq m,
\end{equation*}
also consist all the distinct $q$-cyclotomic cosets modulo $\ell^m$.
That is,
\begin{equation*}
X^{\ell^m}-1=M_{-\rho_0}(X)M_{-\rho_1}(X)M_{-\rho_2}(X)
 \cdots M_{-\rho_e}(X),
\end{equation*}
gives the monic irreducible factorization of $X^{\ell^m}-1$ over $\mathbb F_q$.
Since  $\gamma$ is a primitive fourth  root of unity in $\mathbb F_q$, it follows that
$X^{2\ell^m}+1=(X^{\ell^m}-\gamma)(X^{\ell^m}+\gamma)$. Then
$$
X^{\ell^m}-\gamma=\gamma M_{\rho_0}(\beta X)M_{\rho_1}(\beta X)M_{\rho_2}(\beta X)
 \cdots M_{\rho_e}(\beta X),
$$
$$
X^{\ell^m}+\gamma=-\gamma M_{-\rho_0}(-\beta X)M_{-\rho_1}(-\beta X)M_{-\rho_2}(-\beta X)
 \cdots M_{-\rho_e}(-\beta X),
$$
where $\beta$ is an element in $\mathbb F_q$ with $\beta^{\ell^m}\gamma=1$.
This implies that
$$
X^{2\ell^mp^n}+1=\prod\limits_{i=0}^{e}\hat M_{\rho_i}(\beta X)^{p^n}\hat M_{-\rho_i}(-\beta X)^{p^n},
$$
gives the irreducible factorization of $X^{2\ell^mp^n}+1$ over $\mathbb F_q$.
Recall that
$$
M_{\rho_i}(X)=\prod\limits_{j\in C_{\rho_i}}(X-\eta^j),~1\leq i\leq e.
$$
Now it is routine to check that
$$
\hat M_{\rho_i}^*(\beta X)=\prod\limits_{j\in C_{\rho_i}}(X-\beta^{-1}\eta^j)^*
=\prod\limits_{j\in C_{\rho_i}}(X-\beta\eta^{-j})
=\prod\limits_{j\in C_{-\rho_i}}(X-\beta\eta^{j})
=\hat M_{-\rho_i}(-\beta X).
$$
Using  arguments similar to those of Theorem~\ref{LCDCYCLIC}, we  get  the
statement of {\rm (i)}.

{\rm (ii)}~
From Lemma \ref{lem2}, the monic irreducible factorization of $X^{\ell^m}-1$ can be given as follows:
$$
X^{\ell^m}-1=(X-1)\prod\limits_{r=1}^m\prod\limits_{k=0}^{\frac{\delta(r)}{2}-1}M_{\ell^{m-r}g^k}(X)M_{-\ell^{m-r}g^k}(X).
$$
As discussed previously, an element $\varsigma\in \mathbb{F}_{q^2}$ can be found such that $X^2+1=(X-\varsigma)(X+\varsigma)$,
which gives $X^{2\ell^m}+1=(X^{\ell^m}-\varsigma)(X^{\ell^m}+\varsigma)$.
Further, we have $\theta^{\ell^m}\varsigma=1$, where $\theta\in \mathbb{F}_{q^2}$ satisfies $\theta^{-1}=-\theta=\theta^q$.
Thus,
$$
X^{\ell^m}-\varsigma=
(X-\theta^{-1})\prod\limits_{r=1}^m\prod\limits_{k=0}^{\frac{\delta(r)}{2}-1}\hat M_{\ell^{m-r}g^k}(\theta X)\hat M_{-\ell^{m-r}g^k}(\theta X)
$$
and
$$
X^{\ell^m}+\varsigma=
(X+\theta^{-1})\prod\limits_{r=1}^m\prod\limits_{k=0}^{\frac{\delta(r)}{2}-1}\hat M_{\ell^{m-r}g^k}(-\theta X)\hat M_{-\ell^{m-r}g^k}(-\theta X).
$$
Therefore,
the irreducible factorization
of $X^{2\ell^mp^n}+1$ over $\mathbb F_q$ are given as follows:
$$
X^{2\ell^mp^n}+1=
(X^2+1)^{p^n}\prod\limits_{r=1}^{m}\prod\limits_{k=0}^{\frac{\delta(r)}{2}-1}
I_{r,k}(X)^{p^n}J_{r,k}(X)^{p^n},
$$
where $I_{r,k}(X)=\hat M_{\ell^{m-r}g^k}(\theta X)\hat M_{\ell^{m-r}g^k}(-\theta X)$
and $J_{r,k}(X)=\hat M_{-\ell^{m-r}g^k}(\theta X)\hat M_{-\ell^{m-r}g^k}(-\theta X)$
for every $1\leq r\leq m$
and $0\leq k\leq \frac{\delta(r)}{2}-1$.
We just note that
$$
\hat M_{\ell^{m-r}g^k}^*(\theta X)=\hat M_{-\ell^{m-r}g^k}(-\theta X),~~~
\hat M_{\ell^{m-r}g^k}^*(-\theta X)=\hat M_{-\ell^{m-r}g^k}(\theta X).
$$
That is, $I_{r,k}^*(X)=J_{r,k}(X)$.

Using similar arguments, we obtain (iii) and (iv).
\end{proof}

It is known that self-dual $\lambda$-constacyclic codes can only occur among the classes of cyclic and negacyclic codes, i.e., $\lambda = 1$ or $-1$ (e.g. \cite{Dinh2014}). It is also known that self-dual cyclic codes  over a finite field exist if and only if
the code length is even and the characteristic of the underlying field is two  (\cite{Jia}).
Thus, we focus  on self-dual negacyclic codes, which  also have received a good deal of
attention.

\begin{Corollary} With the notations as in Theorem \ref{LCDNEGACYCLIC}, we have that

{\rm (i)}~if $q\equiv1~(\bmod~4)$, then there are exactly $(p^n+1)^{e+1}$ self-dual negacyclic codes of length $2\ell^mp^n$
over $\mathbb F_q$ generated by
$$
\prod\limits_{i=0}^{e}\hat M_{\rho_i}(\beta X)^{\varepsilon_i}\hat M_{-\rho_i}(-\beta X)^{p^n-\varepsilon_i},
~~~0\leq\varepsilon_i\leq p^n,~~~i=0,1,\cdots,e;
$$

{\rm (ii)}~if $q\equiv3~(\bmod~4)$, then there does not exist self-dual negacyclic codes of length $2\ell^mp^n$
over $\mathbb F_q$.
\end{Corollary}
\begin{proof}
From Theorem \ref{LCDNEGACYCLIC} (i) and its proof, we know that
$$
X^{2\ell^mp^n}+1=\prod\limits_{i=0}^{e}\hat M_{\rho_i}(\beta X)^{p^n}\hat M_{-\rho_i}(-\beta X)^{p^n},
$$
gives the irreducible factorization of $X^{2\ell^mp^n}+1$ over $\mathbb F_q$.
Moreover, $\hat M_{\rho_i}^*(\beta X)
=\hat M_{-\rho_i}(-\beta X)$.
We deduce that (i) holds true.

(ii)~It follows from \ref{LCDNEGACYCLIC} (ii)-(iv) that $X^2+1$ is a self-reciprocal irreducible polynomial of $X^{2\ell^mp^n}+1$
over $\mathbb{F}_q$. This gives that there does not exist self-dual negacyclic codes of length $2\ell^mp^n$
over $\mathbb F_q$.

\end{proof}

\end{document}